\documentclass{article}
\usepackage{amsfonts}

\begin{document}

\newtheorem{theoremm}{Theorem}
\newtheorem{conditionss}{Condition}
\newtheorem{thesiss}[theoremm]{Thesis}
\newtheorem{definitionn}[theoremm]{Definition}
\newtheorem{lemmaa}[theoremm]{Lemma}
\newtheorem{notationn}[theoremm]{Notation}
\newtheorem{propositionn}[theoremm]{Proposition}
\newtheorem{conventionn}[theoremm]{Convention}
\newtheorem{examplee}[theoremm]{Example}
\newtheorem{remarkk}[theoremm]{Remark}
\newtheorem{factt}[theoremm]{Fact}
\newtheorem{exercisee}[theoremm]{Exercise}
\newtheorem{questionn}[theoremm]{Open Problem}
\newtheorem{conjecturee}[theoremm]{Conjecture}

\newenvironment{exercise}{\begin{exercisee} \em}{ \end{exercisee}}
\newenvironment{definition}{\begin{definitionn} \em}{ \end{definitionn}}
\newenvironment{theorem}{\begin{theoremm}}{\end{theoremm}}
\newenvironment{lemma}{\begin{lemmaa}}{\end{lemmaa}}
\newenvironment{proposition}{\begin{propositionn} }{\end{propositionn}}
\newenvironment{convention}{\begin{conventionn} \em}{\end{conventionn}}
\newenvironment{remark}{\begin{remarkk} \em}{\end{remarkk}}
\newenvironment{proof}{ {\bf Proof.} }{\  \rule{2.5mm}{2.5mm} \vspace{.2in} }
\newenvironment{idea}{ {\bf Idea.} }{\  \rule{1.5mm}{1.5mm} \vspace{.15in} }

\newenvironment{fact}{\begin{factt}}{\end{factt}}
\newenvironment{notation}{\begin{notationn} \em}{\end{notationn}}
\newenvironment{conditions}{\begin{conditionss} \em}{\end{conditionss}}
\newenvironment{question}{\begin{questionn}}{\end{questionn}}
\newenvironment{conjecture}{\begin{conjecturee}}{\end{conjecturee}}

\title{Soundness and completeness of the cirquent calculus system CL6 for computability logic$^*$\footnotetext[1]{This work was supported by the NNSF (60974082) of China.}}

\author{{\em Wenyan Xu \ and \ Sanyang Liu}\\
{\em\small Department of Mathematics, Xidian University, Xi'an, 710071, PR China}}
\date{}
\maketitle
\begin{abstract}
Computability logic is a formal theory of computability. The earlier
article ``Introduction to cirquent calculus and abstract resource
semantics" by Japaridze proved soundness and completeness for the basic fragment
{\bf CL5} of computability logic. The present article extends that
result to the more expressive cirquent calculus system {\bf CL6},
which is a conservative extension of both {\bf CL5} and classical propositional logic.
\end{abstract}


\section{Introduction}

Computability logic(CoL), introduced by G. Japaridze
\cite{Key3}-\cite{Key5}, is a semantical and mathematical platform
for redeveloping logic as a formal theory of computability. Formulas in CoL represent interactive computational problems,
understood as games between a machine and its environment (symbolically named as $\top$ and $\bot$, respectively); logical operators stand for operations on such problems; ``truth" of a problem/game means existence of an algorithmic solution, i.e. $\top$'s effective winning strategy; and validity of a logical formula is understood as such truth under every particular interpretation of atoms. The approach induces a rich collection of (old or new) logical operators. Among those, relevant to this paper are $\neg$ (negation), $\vee$ (parallel disjunction) and $\wedge$ (parallel conjunction). Intuitively, $\neg$ is a role switch operator: $\neg A$ is the game $A$ with the roles of $\top$ and $\bot$ interchanged ($\top$'s legal moves and wins become those of $\bot$, and vice versa). Both $A\wedge B$ and $A\vee B$ are games playing which means playing the two components $A$ and $B$ simultaneously (in parallel). In $A\wedge B$, $\top$ is the winner if it wins in both components, while in $A\vee B$ winning in just one component is sufficient. The symbols $\top$ and $\bot$, together with denoting the two players, are also used to denote two special (the simplest) sorts of games. Namely, $\top$ is a moveless (``elementary") game automatically won by the player $\top$, and $\bot$ is a moveless game automatically won by $\bot$.

Cirquent calculus is a
refinement of sequent calculus. Unlike the more traditional proof theories that manipulate tree-like objects (formulas, sequents, hypersequents, etc.), cirquent calculus deals with graph-style structures termed {\em cirquents}, with its main characteristic feature thus being allowing to explicitly account for {\em sharing} subcomponents between different subcomponents. The approach was introduced by Japaridze
\cite{Key1} as a new deductive tool for CoL and was developed later in \cite{Key10}-\cite{Key12}. The paper \cite{Key1} constructed a cirquent calculus system {\bf CL5} for the basic
$(\neg,\wedge,\vee)$-fragment of CoL, and proved its soundness and completeness with respect to the semantics of CoL.

The atoms of {\bf CL5} represent computational problems in general, and are said to be {\em general atoms}. The so called {\em elementary atoms},
representing computational problems of zero degree of interactivity (such as the earlier-mentioned games $\top$ and $\bot$)
and studied in other pieces of literature on CoL, are not among them.
Thus, {\bf CL5} only describes valid
computability principles for general problems. This is a significant
limitation of expressive power. For example, the problem
$A\rightarrow A{\wedge}A$ is not valid in CoL when $A$ is a general atom, but becomes valid (as any classical tautology for that matter) when $A$ is elementary. So the language of {\bf CL5} naturally calls for an extension.

Japaridze \cite{Key1} claimed without a proof that the soundness
and completeness result for {\bf CL5} could be extended to the more expressive
cirquent calculus system {\bf CL6} (reproduced later), which is a conservative
extension of both {\bf CL5} and classical propositional logic.
This article is devoted to a
soundness and completeness proof for system {\bf CL6}, thus contributing
to the task of extending the cirquent-calculus approach
so as to accommodate incrementally expressive fragments of CoL.

\section{Preliminaries}
This paper primarily targets readers already familiar with Japaridze \cite{Key1}, and can essentially be treated as a technical appendix to the latter. However, in order to make it reasonably self-contained, in this section we reproduce the basic concepts from \cite{Key1} on which the later parts of the paper will rely. An interested reader may consult \cite{Key1} for additional explanations, illustrations and examples.

The language of {\bf CL6} is more expressive than that of {\bf CL5} in that,
 along with
the old atoms of {\bf CL5} called {\bf general}, it has an
additional sort of atoms called {\bf elementary}, including
{\bf non-logical} elementary atoms and {\bf logical} atoms $\top$ and $\bot$. On the other hand, all general atoms are non-logical. We
use the uppercase letters $P, Q, R, S$ as metavariables
for general atoms, and the lowercase $p, q, r, s$ as
metavariables for non-logical elementary atoms.
A {\bf CL6-formula} is built from atoms in the standard way
using the connectives $\neg$,$\vee$,$\wedge$, with $F\rightarrow G$ understood as an
abbreviation for $\neg F\vee G$ and $\neg$ limited only to non-logical atoms, where $\neg\neg F$ is understood as $F$, $\neg(F\wedge G)$ as $\neg F\vee\neg G$, $\neg(F\vee G)$ as $\neg F\wedge\neg G$, $\neg\top$ as $\bot$, and $\neg\bot$ as $\top$.
An atom $P$ (resp. $p$) and its {\bf negation} $\neg P$ (resp. $\neg p$) is called a {\bf literal},
and the two literals are said to be {\bf opposite}. A {\bf CL6}-formula is said to be {\bf elementary} iff it does not contain general
atoms. Throughout the rest of this paper, unless otherwise specified, by an
``atom'' or a ``formula'' we mean one of the language of {\bf CL6}.

Where $k\geq 0$, a $k-$ary {\bf pool} is a sequence
$\langle F_1, F_2, \ldots, F_k\rangle$ of $k$ formulas.  Since we may have $F_i=F_j$ for some $i\neq j$ in such a sequence, we use the term {\bf oformula} to refer to a formula together with a particular occurrence of it in the pool.   For example, the pool $\langle E,F,G,E\rangle $ has three formulas but four oformulas.
Similarly, the terms ``{\bf oliteral}",``{\bf oatom}", etc. will be used
in this paper to refer to the corresponding entities together with particular occurrences.
A $k-$ary {\bf structure} is a finite sequence {\bf St}$=\langle \Gamma_1,\ldots,\Gamma_m\rangle $, where $m\geq 0$ and each $\Gamma_i$, said to be a {\bf group} of {\bf St}, is a subset of $\{1,\ldots,k\}$. Again, to differentiate
between a group as  such and a particular occurrence of a group in the structure, we use the term {\bf ogroup} for the latter. For example, the structure $\langle\{2,3\},\{2,3\},\{1,4\},\emptyset\rangle$ has three groups but four ogroups.

A $k$-ary ($k\geq 0$) {\bf cirquent} is a pair $C=({\bf St}^{C},{\bf Pl}^{C})$, where ${\bf St}^{C}$, called the {\bf structure} of $C$, is a $k$-ary structure, and ${\bf Pl}^{C}$, called the {\bf pool} of $C$, is a $k$-ary pool. An ogroup of such a $C$ will mean an ogroup of ${St}^{C}$, and an oformula of $C$ will mean an oformula of ${\bf Pl}^{C}$. Usually, we understand the groups of a cirquent as sets of its oformulas rather than sets of the corresponding ordinal numbers. Thus, if ${\bf Pl}^{C}=\langle E,F,G,E\rangle$ and $\Gamma=\{2,4\}$, we would think of $\Gamma$ simply as the set $\{F,E\}$, and say that $\Gamma$ {\bf contains} $F$ and $E$. When both the pool and the structure of a cirquent $C$ are empty, i.e. $C=(\langle\rangle,\langle\rangle)$, we call it the {\bf empty cirquent}.

Rather than writing cirquents as ordered tuples in the above-described style, we prefer to represent them through (and identify them with) {\bf diagrams}. Below is such a representation for the cirquent whose pool is $\langle E,F,G,H\rangle$ and whose structure is $\langle \{1,2\},\{2\},\{3,4\}\rangle$.

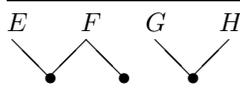
\begin{center}
\begin{picture}(80,60)(0,7)
\put(-10,50){\line(1,0){90}}\put(-10,38){$E\ \ \ \ \ \ F\ \ \ \ \  G\ \ \ \ \ \ H$}
\put(-8,35){\line(1,-1){14}}\put(20,35){\line(-1,-1){14}}\put(4,18){$\bullet$}
\put(20,35){\line(1,-1){14}}\put(46,35){\line(1,-1){14}}\put(32,18){$\bullet$}
\put(74,35){\line(-1,-1){14}}\put(58,18){$\bullet$}
\end{picture}
\end{center}
The top level of a diagram thus indicates the oformulas of the cirquent, and the bottom level gives its ogroups. An ogroup $\Gamma$ is represented by a $\bullet$, and the lines connecting $\Gamma$ with oformulas, called {\bf arcs}, are pointing to the oformulas that $\Gamma$ contains. Finally, we put a horizontal line at the top of the diagram to indicate that this is one cirquent rather than two or more cirquents put together.

A {\bf model} is a function $M$ that assigns a truth value
--- {\itshape true} (1) or {\itshape false} (0) --- to each atom, with $\top$ being
always assigned {\itshape true} and $\bot$ {\itshape false}, and extends
to compound formulas in the standard classical way. Let $M$ be a model, and
$C$ a cirquent. We say that a group $\Gamma$ of $C$ is {\bf true} in $M$
iff at least one of its oformulas is so. And $C$ is {\bf true} in $M$ if
every group of $C$ is so. Otherwise, $C$ is {\bf false}. Finally,
$C$ or a group $\Gamma$ of it is a {\bf tautology} iff it is true in every model.

A {\bf substitution} is a function $\sigma$ that sends every general atom $P$ to
some formula $\sigma(P)$, and sends every elementary atom to itself.
If, (for every general atom $P$), such a $\sigma(P)$ is an atom, then $\sigma$ is said to be
an {\bf atomic-level substitution}.

Let $A$ and $B$ be cirquents. We say that $B$ is an {\bf instance} of $A$ iff
$B = \sigma(A)$ for some substitution $\sigma$, where $\sigma(A)$ is the result of replacing
in all oformulas of $A$ every (general or elementary) atom $\alpha$ by $\sigma(\alpha)$; and $B$ is an {\bf atomic-level instance}
of $A$ iff $B = \sigma(A)$ for some atomic-level substitution $\sigma$.

A cirquent is said to be {\bf binary} iff no general atom has more than two occurrences
in it. A binary cirquent is said to be {\bf normal} iff, whenever it has two occurrences of
a general atom, one occurrence is negative and the other is positive.
A {\bf binary tautology} (resp. {\bf normal binary tautology}) is a binary (resp. normal binary)
cirquent that is a tautology.

The set of rules of
{\bf CL6} is obtained from that of {\bf CL5} by adding to it $\top$
as an additional axiom, plus the rule of contraction limited only to
elementary formulas. Below we reproduce those rules from \cite{Key1}, followed by illustrations.

{\bf Axioms (A):}\ \ \ Axioms are ``rules" with no premises. There are three sorts of axioms in {\bf CL6}.
The first one is the empty cirquent. The second one is any cirquent that has exactly two oformulas $F$ and $\neg F$, for
some arbitrary formula $F$, and an ogroup that contains $F$ and $\neg F$. In other words, this is the cirquent $(\langle \{1,2\}\rangle,\langle F,\neg F\rangle)$.
The third one is a cirquent that has exactly one oformula $\top$ and one ogroup that contains $\top$, i.e. the cirquent $(\langle\{1\}\rangle,\langle \top\rangle)$.

{\bf Mix (M):}\ \ \ According to this rule, the conclusion
can be obtained by simply putting any two cirquents (premises) together, thus creating one cirquent out of two.

{\bf Exchange (E):}\ \ \ This rule comes in two versions: {\bf oformula exchange} and {\bf ogroup exchange}.
The conclusion of oformula exchange is obtained by interchanging in the premise two adjacent oformulas $E$ and $F$, and
redirecting to $E$ (resp. $F$) all arcs that were originally pointing to $E$ (resp. $F$). Ogroup exchange is the same, with the only difference that the objects interchanged are ogroups.

{\bf Weakening (W):}\ \ \ This rule also comes in two versions: {\bf ogroup weakening} and {\bf pool weakening}.
A conclusion of ogroup weakening is obtained by adding in the premise a new arc between an existing
ogroup and an existing oformula. As for pool weakening, a conclusion is obtained through
inserting a new oformula anywhere in the pool of the premise.

{\bf Duplication (D):}\ \ \ A conclusion of this rule is obtained by replacing in the premise some ogroup $\Gamma$ by two adjacent ogroups that, as groups, are identical with $\Gamma$.

{\bf Contraction (C):}\ \ \ According to this rule, if a cirquent (a premise) has two adjacent elementary oformulas $F$ (the first), $F$ (the second) that are identical,
then a conclusion can be obtained by merging $F$,$F$ into $F$ and redirecting to the latter all arcs that were originally pointing to the first or the second $F$.

{\bf $\vee-$introduction ($\vee$):}\ \ \ For the convenience of description, we explain this rule in the bottom-up view.
According to this rule, if a cirquent (the conclusion) has an oformula $E\vee F$ that is contained by at least one ogroup, then the premise can be obtained by splitting the original $E\vee F$ into two adjacent oformulas $E$ and $F$, and redirecting to {\em both} $E$ and $F$ all arcs that were originally pointing to $A\vee B$.

{\bf $\wedge-$introduction ($\wedge$):}\ \ \ This rule, again, is more conveniently described in the bottom-up view. According to this rule, if a cirquent (the conclusion) has an oformula $E\wedge F$ that is contained by at least one ogroup, then the premise
can be obtained by splitting the original $E\wedge F$ into two adjacent oformulas $E$ and $F$, and splitting every ogroup $\Gamma$ that originally contained $E\wedge F$ into two adjacent ogroups $\Gamma^{E}$ and $\Gamma^{F}$, where $\Gamma^{E}$ contains $E$ (but not $F$), and $\Gamma^{F}$ contains $F$ (but not $E$), with all other ($\neq E\wedge F$) oformulas of $\Gamma$ contained by both $\Gamma^{E}$ and $\Gamma^{F}$.

Below we provide illustrations for all rules, in each case an abbreviated name of the rule standing next to the horizontal line separating the premises from the conclusions. Our illustrations for the axioms (the ``{\bf A}" labeled rules) are specific cirquents or schemate of such; our illustrations for all other rules are merely examples chosen arbitrarily. Unfortunately, no systematic ways for schematically representing cirquent calculus rules have been elaborated so far. This explains why we appeal to examples instead.

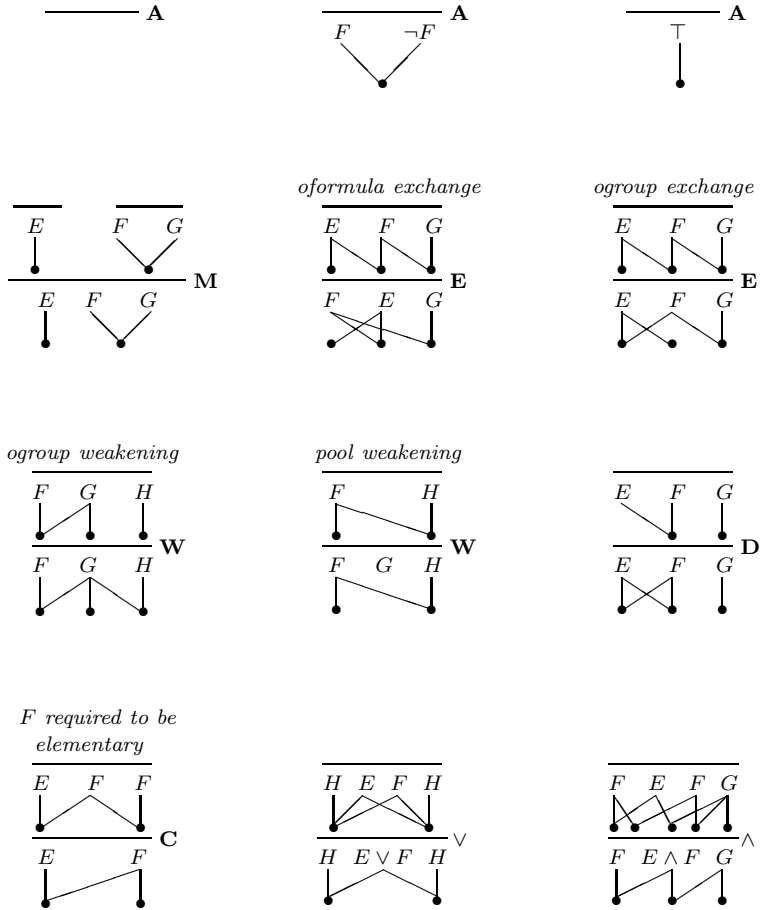
\begin{center}
\begin{picture}(80,60)(0,7)\footnotesize
\put(-90,50){\line(1,0){35}}\put(-52,47){\bf A}
\put(15,50){\line(1,0){45}}\put(19,40){$F\ \ \ \ \ \ \ \neg F$}\put(63,47){\bf A}
\put(22,38){\line(1,-1){15}}\put(52,38){\line(-1,-1){15}}\put(35.5,21){$\bullet$}
\put(130,50){\line(1,0){35}}\put(146,40){$\top$}\put(168,47){\bf A}
\put(150,38){\line(0,-1){15}}\put(148,21){$\bullet$}
\end{picture}
\end{center}

\begin{center}
\begin{picture}(80,80)(0,20)\footnotesize
\put(15,0){\begin{picture}(80,80)\put(-117,80){\line(1,0){18}}\put(-78,80){\line(1,0){25}}
\put(-112,70){$E$}\put(-109,68){\line(0,-1){11}}\put(-111,54){$\bullet$}
\put(-119,52){\line(1,0){67}}\put(-49,49){\bf M}
\put(-80,70){$F\ \ \ \ \ G$}\put(-78,68){\line(1,-1){11}}\put(-55,68){\line(-1,-1){11}}\put(-68,54){$\bullet$}
\put(-108,42){$E$}\put(-105,40){\line(0,-1){11}}\put(-107,26){$\bullet$}
\put(-90,42){$F\ \ \ \ \ G$}\put(-88,40){\line(1,-1){11}}\put(-65,40){\line(-1,-1){11}}\put(-78.5,26){$\bullet$}\end{picture}}

\put(5,85){\it oformula  exchange}
\put(15,80){\line(1,0){45}}\put(15,70){$E\ \ \ \ \ F\ \ \ \ G$}
\put(18,68){\line(0,-1){11}}\put(16,54){$\bullet$}
\put(18,68){\line(3,-2){18}}\put(37,68){\line(0,-1){11}}\put(35,54){$\bullet$}
\put(37,68){\line(3,-2){18}}\put(56,68){\line(0,-1){11}}\put(54,54){$\bullet$}
\put(15,52){\line(1,0){45}}\put(63,49){\bf E}
\put(15,42){$F\ \ \ \ \ E\ \ \ \ G$}
\put(16,26){$\bullet$}\put(35,26){$\bullet$}\put(54,26){$\bullet$}
\put(37,40){\line(-3,-2){18}}\put(37,40){\line(0,-1){11}}
\put(18,40){\line(3,-2){18}}\put(56,40){\line(0,-1){11}}\put(18,40){\line(3,-1){38}}

\put(0,0){\begin{picture}(80,80)\put(117,85){\it ogroup  exchange}
\put(125,80){\line(1,0){45}}\put(125,70){$E\ \ \ \ \ F\ \ \ \ G$}
\put(128,68){\line(0,-1){11}}\put(126,54){$\bullet$}
\put(128,68){\line(3,-2){18}}\put(147,68){\line(0,-1){11}}\put(145,54){$\bullet$}
\put(147,68){\line(3,-2){18}}\put(166,68){\line(0,-1){11}}\put(164,54){$\bullet$}
\put(125,52){\line(1,0){45}}\put(173,49){\bf E}
\put(125,42){$E\ \ \ \ \ F\ \ \ \ G$}
\put(126,26){$\bullet$}\put(145,26){$\bullet$}\put(164,26){$\bullet$}
\put(128,40){\line(0,-1){11}}\put(147,40){\line(-3,-2){18}}
\put(128,40){\line(3,-2){18}}\put(147,40){\line(3,-2){18}}\put(166,40){\line(0,-1){11}}\end{picture}}
\end{picture}
\end{center}

\begin{center}
\begin{picture}(80,80)(0,30)\footnotesize
\put(-105,85){\it ogroup weakening}
\put(0,0){\begin{picture}(80,80)\put(-95,80){\line(1,0){45}}
\put(-95,70){$F\ \ \ \  G\ \ \ \ \ H$}\put(-92,68){\line(0,-1){11}}
\put(-94,54){$\bullet$}\put(-75,54){$\bullet$}
\put(-73,68){\line(-3,-2){18}}\put(-73,68){\line(0,-1){11}}\put(-53,68){\line(0,-1){11}}\put(-55,54){$\bullet$}
\put(-95,52){\line(1,0){45}}\put(-47,49){\bf W}
\put(-95,42){$F\ \ \ \ G\ \ \ \ \ H$}\put(-92,40){\line(0,-1){12}}
\put(-94,25){$\bullet$}\put(-75,25){$\bullet$}
\put(-73,40){\line(-3,-2){18}}\put(-73,40){\line(0,-1){12}}\put(-53,40){\line(0,-1){12}}\put(-55,25){$\bullet$}
\put(-73,40){\line(3,-2){19}}\end{picture}}

\put(12,85){\it pool weakening}
\put(110,0){\begin{picture}(80,80)\put(-95,80){\line(1,0){45}}
\put(-93,70){$F\ \ \ \ \ \ \ \ \ \ H$}
\put(-90,68){\line(0,-1){11}}
\put(-92,54){$\bullet$}\put(-90,68){\line(3,-1){37}}
\put(-54,68){\line(0,-1){11}}\put(-56,54){$\bullet$}
\put(-95,52){\line(1,0){45}}\put(-47,49){\bf W}
\put(-93,42){$F\ \ \ \ G\ \ \ \  H$}
\put(-90,40){\line(0,-1){11}}
\put(-92,26){$\bullet$}\put(-90,40){\line(3,-1){37}}
\put(-54,40){\line(0,-1){11}}\put(-56,26){$\bullet$}
\end{picture}}

\put(0,0){\begin{picture}(80,80)
\put(125,80){\line(1,0){45}}\put(125,70){$E\ \ \ \ \ F\ \ \ \ G$}
\put(128,68){\line(3,-2){18}}\put(147,68){\line(0,-1){11}}\put(145,54){$\bullet$}
\put(166,68){\line(0,-1){11}}\put(164,54){$\bullet$}
\put(125,52){\line(1,0){45}}\put(173,49){\bf D}
\put(125,42){$E\ \ \ \ \ F\ \ \ \ G$}
\put(126,26){$\bullet$}\put(145,26){$\bullet$}\put(164,26){$\bullet$}
\put(128,40){\line(0,-1){11}}\put(147,40){\line(-3,-2){18}}
\put(128,40){\line(3,-2){18}}\put(147,40){\line(0,-1){11}}\put(166,40){\line(0,-1){11}}\end{picture}}
\end{picture}
\end{center}

\begin{center}
\begin{picture}(80,120)(0,10)\footnotesize

\put(0,0){\begin{picture}(80,80)
\put(-100,95){\it $F$ required to be}
\put(-94,85){\it elementary}
\put(-95,80){\line(1,0){45}}
\put(-95,70){$E\ \ \ \ \ F\ \ \ \ F$}\put(-92,68){\line(0,-1){11}}
\put(-94,54){$\bullet$}
\put(-73,68){\line(-3,-2){18}}\put(-73,68){\line(3,-2){19}}\put(-54,68){\line(0,-1){11}}\put(-56,54){$\bullet$}
\put(-95,52){\line(1,0){45}}\put(-47,49){\bf C}
\put(-93,42){$E\ \ \ \ \ \ \ \ \ \ F$}\put(-90,40){\line(0,-1){12}}
\put(-92,25){$\bullet$}\put(-54,40){\line(0,-1){12}}\put(-56,25){$\bullet$}
\put(-54,40){\line(-3,-1){36}}\end{picture}}

\put(110,0){\begin{picture}(80,80)\put(-95,80){\line(1,0){45}}
\put(-95,70){$H\ \  E \ \  F\ \  H$}
\put(-91,68){\line(0,-1){11}}
\put(-93,54){$\bullet$}\put(-80,68){\line(-1,-1){11}}\put(-67,68){\line(-2,-1){23}}
\put(-80,68){\line(2,-1){24}}\put(-67,68){\line(1,-1){12}}
\put(-55,68){\line(0,-1){11}}\put(-57,54){$\bullet$}
\put(-97,52){\line(1,0){49}}\put(-46,49){\bf $\vee$}
\put(-97,42){$H \ \ E\vee F\ \ H$}
\put(-93,40){\line(0,-1){11}}
\put(-95,26){$\bullet$}\put(-72,40){\line(-2,-1){21}}\put(-72,40){\line(2,-1){20}}
\put(-52,40){\line(0,-1){11}}\put(-54,26){$\bullet$}
\end{picture}}

\put(0,0){\begin{picture}(80,80)
\put(123,80){\line(1,0){49}}\put(123,70){$F\ \ \ E \ \ \ F\ \ G$}
\put(141,68){\line(-3,-2){17}}\put(125,68){\line(2,-3){8}}
\put(156,68){\line(-2,-1){22}}\put(141,68){\line(1,-2){6}}\put(168,68){\line(-2,-1){21}}
\put(156,68){\line(0,-1){11}}\put(168,68){\line(-1,-1){11}}
\put(125,68){\line(0,-1){11}}\put(123,54){$\bullet$}
\put(131,54){$\bullet$}\put(145,54){$\bullet$}\put(154,54){$\bullet$}
\put(168,68){\line(0,-1){11}}\put(166,54){$\bullet$}

\put(123,52){\line(1,0){49}}\put(173,49){\bf $\wedge$}
\put(123,42){$F \ \ E\wedge F\ \ G$}
\put(145,26){$\bullet$}\put(164,26){$\bullet$}
\put(147,40){\line(0,-1){11}}\put(166,40){\line(-3,-2){18}}
\put(126,40){\line(0,-1){11}}\put(124,26){$\bullet$}\put(147,40){\line(-2,-1){21}}
\put(166,40){\line(0,-1){11}}\end{picture}}
\end{picture}
\end{center}

The above are all eight rules of {\bf CL6}. As a warm-up exercise, the reader may try to verify that {\bf CL6} proves
$p\rightarrow p\wedge p$ but does not prove
$P\rightarrow P\wedge P$.

As an aside, the earlier mentioned system {\bf CL5} differs from {\bf CL6} in that the $\top-$axiom and the contraction rules are absent there. Also, as noted, the language of {\bf CL5} does not allow elementary atoms. In next section we will see that our proofs are carried out purely syntactically, based on the soundness and completeness of system {\bf CL2}
(introduced in Japaridze \cite{Key2}) with respect to the semantics of CoL. That is to say we do not directly use the semantics of CoL. So,
below we only explain what the language of {\bf CL2} and its rules are, without providing any formal definitions (on top of the brief informal explanations given in Section 1) of the underlying CoL semantics. If necessary, such definitions can be found in  \cite{Key5}.

The language of {\bf CL2} is more expressive than the one in which formulas  of {\bf CL6} are written because, on top of $\neg$,$\vee$,$\wedge$, it has the binary connectives $\sqcap$ and $\sqcup$, called {\it choice operators}. The {\bf CL2-formulas} are built from atoms (including general atoms and elementary atoms) in the standard way using the connectives $\neg$,$\vee$,$\wedge$,$\sqcap$,$\sqcup$.
As in the case of {\bf CL6}-formulas, the operator $\neg$ is only allowed to be applied to non-logical atoms.
A {\bf CL2}-formula is said to be {\bf elementary} iff it contains neither general atoms nor $\sqcap$,$\sqcup$.
A {\bf positive occurrence} (resp. {\bf negative occurrence}) of an atom is one that
is not (resp. is) in the scope of $\neg$.
A {\bf surface occurrence} of a subformula of a {\bf CL2}-formula is an occurrence that is not in the scope of $\sqcap$,$\sqcup$.
A {\bf general literal} is $P$ or $\neg P$, where $P$ is a general atom.
The {\bf elementarization} of a {\bf CL2}-formula $A$ is the result of replacing in $A$ every positive surface occurrence of each general literal by $\bot$, every surface occurrence of each $\sqcup-$subformula by $\bot$, and every surface occurrence of each $\sqcap-$subformula by $\top$. A {\bf CL2}-formula is said to be {\bf stable} iff its
elementarization is a tautology of classical logic.

{\bf CL2} has the following three inference rules.

{\bf Rule (a):}\ \ \ $\overrightarrow{H} \mapsto F$, where $F$ is stable and $\overrightarrow{H}$ is the smallest set of formulas such that, whenever $F$ has a surface occurrence of a subformula $G_1\sqcap G_2$, for both $i$$\in\{$1,2$\}$, $\overrightarrow{H}$ contains the result of replacing that occurrence in $F$ by $G_i$.

{\bf Rule (b):}\ \ \ $H\mapsto F$, where $H$ is the result of replacing in $F$ a surface occurrence of a subformula $G_1\sqcup G_2$ by $G_1$ or $G_2$.

{\bf Rule (c):}\ \ \ $H\mapsto F$, where $H$ is the result of replacing in $F$ two --- one positive and one negative --- surface occurrences of some general atom by a non-logical elementary atom that does not occur in $F$.

The set $\overrightarrow{H}$ of the premises of Rule {\bf (a)} may be empty, in which case the rule (its conclusion, that is) acts like an axiom. Otherwise, the system has no (other) axioms.

\section{Soundness and completeness of CL6}

In what follows, we may use names such as (AME) to refer to the subsystem of {\bf CL6} consisting only of the rules whose names are listed between the parentheses. So, (AME) refers to the system that only has axioms, exchange and mix. The same notation can be used next to the horizontal line separating two cirquents to indicate that the lower cirquent (``conclusion") can be obtained from the upper cirquent (``premise") by whatever number of applications of the corresponding rules. The following Lemmas 1, 2, 3, 4 are precisely Lemmas 4, 5, 10 and 11 of \cite{Key1}, so we state them without proofs (such proofs are given in \cite{Key1}).

\begin{lemma}
All of the rules of {\bf CL6} preserve truth in the top-down direction. Taking no premises, (the conclusion of)
axioms are thus tautologies.
\end{lemma}

\begin{lemma}
The rules of mix, exchange, duplication, contraction, $\vee$-introduction and $\wedge$-introduction
preserve truth in the bottom-up direction as well.
\end{lemma}

\begin{lemma}
The rules of mix, exchange, duplication, $\vee$-introduction and $\wedge$-introduction preserve binarity and normal binarity in
both top-down and bottom-up directions.
\end{lemma}

\begin{lemma}
Weakening preserves binarity and normal binarity in the bottom-up direction.
\end{lemma}

\begin{lemma}
If {\bf CL6} proves a cirquent $C$, then it also proves every instance of $C$.
\end{lemma}

\begin{proof}
Let $T$ be a proof tree of an arbitrary cirquent $C$, $C'$ be an arbitrary instance
of $C$, and $\sigma$ be a substitution with $\sigma(C)=C'$. Replace every oformula
$F$ of every cirquent of $T$ by $\sigma(F)$. It is not hard to see that the resulting
tree $T'$, which uses exactly the same rules as $T$ does, is a proof of $C'$.
\end{proof}

\begin{lemma}
Contraction preserves binarity and normal binarity
in both top-down and bottom-up directions.
\end{lemma}
\begin{proof}
This is so because contraction limited to elementary formulas
can never affect what general atoms occur in a cirquent and how many times they occur.
\end{proof}

\begin{lemma}
A cirquent is provable in {\bf CL6} iff it is an instance of a
binary tautology.
\end{lemma}
\begin{proof}
$(\Rightarrow)$ Consider an arbitrary cirquent $A$ provable in {\bf
CL6}. By induction on the height of its proof tree, we want to show
that $A$ is an instance of a binary tautology.

The above is obvious when $A$ is an axiom.

Suppose now $A$ is derived by exchange from $B$. Let us just
consider oformula exchange, with ogroup exchange being similar. By
the induction hypothesis, $B$ is an instance of a binary tautology $B'$.
Let $A'$ be the result of applying exchange to $B'$ ``at the
same place" as it was applied to $B$ when deriving $A$ from it, as
illustrated in the following example:
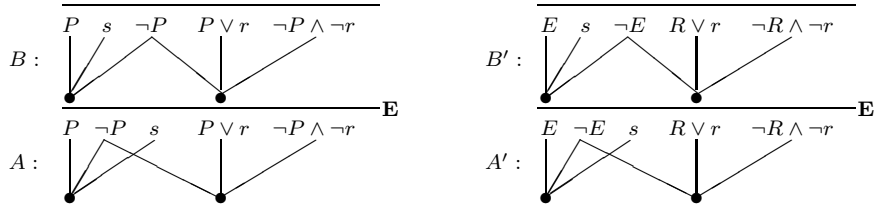
\begin{center}
\begin{picture}(80,90)(0,50)\footnotesize
\put(0,0){\begin{picture}(80,90)\put(-120,130){\line(1,0){120}}
\put(-120,120){$P\ \ \ s\ \ \ \neg P\ \ \ \ P\vee r\ \ \ \neg
P\wedge\neg r$}

\put(-117,118){\line(0,-1){21}}\put(-104,118){\line(-3,-5){13}}\put(-86,118){\line(-4,-3){32}}

\put(-86,118){\line(5,-4){27}}\put(-60,118){\line(0,-1){20}}\put(-24,118){\line(-5,-3){35}}\put(-117,95){\circle*{4}}\put(-60,95){\circle*{4}}

\put(-120,91){\line(1,0){120}}\put(1,88){${\bf E}$}

\put(-120,81){$P\ \ \neg P\ \ \ s\ \ \ \ \ P\vee r\ \ \ \neg
P\wedge\neg r$}

\put(-117,79){\line(0,-1){21}}\put(-104,79){\line(-3,-5){13}}\put(-85,79){\line(-3,-2){32}}

\put(-104,79){\line(2,-1){43}}\put(-60,79){\line(0,-1){20}}\put(-24,79){\line(-5,-3){37}}\put(-117,57){\circle*{4}}\put(-60,57){\circle*{4}}

\put(-140,107){$B:$}\put(-140,68){$A:$}\end{picture}}

\put(180,0){\begin{picture}(80,90)\put(-120,130){\line(1,0){120}}
\put(-119,120){$ E\ \ \ s\ \ \ \neg E\ \ \ R\vee r\ \ \ \ \neg
R\wedge\neg r$}

\put(-117,118){\line(0,-1){21}}\put(-104,118){\line(-3,-5){13}}\put(-86,118){\line(-4,-3){32}}

\put(-86,118){\line(5,-4){27}}\put(-60,118){\line(0,-1){20}}\put(-24,118){\line(-5,-3){35}}\put(-117,95){\circle*{4}}\put(-60,95){\circle*{4}}

\put(-120,91){\line(1,0){120}}\put(1,88){${\bf E}$}

\put(-119,81){$ E\ \ \neg E\ \ \ s\ \ \ \ R\vee r\ \ \ \ \neg
R\wedge\neg r$}

\put(-117,79){\line(0,-1){21}}\put(-104,79){\line(-3,-5){13}}\put(-85,79){\line(-3,-2){32}}

\put(-104,79){\line(2,-1){43}}\put(-60,79){\line(0,-1){20}}\put(-24,79){\line(-5,-3){37}}\put(-117,57){\circle*{4}}\put(-60,57){\circle*{4}}

\put(-140,107){$B':$}\put(-140,68){$A':$}\end{picture}}

\end{picture}
\end{center}
Obviously $A$ will be an instance of $A'$. It remains to note that,
by Lemmas 1 and 3, $A'$ is a binary tautology.

The rules of duplication, $\vee$-introduction and
$\wedge$-introduction can be handled in a similar way.

Next, suppose $A$ is derived from $B$ and $C$ by mix. By the
induction hypothesis, $B$ and $C$ are instances of some binary tautologies
$B'$ and $C'$, respectively. We
may assume that no general atom $P$ occurs in
both $B'$ and $C'$, for otherwise, in one of the cirquents, rename
$P$ into another general atom $Q$ different from everything else. Let $A'$ be the
result of applying mix to $B'$ and $C'$. By Lemmas 1 and 3,
$A'$ is a binary tautology. And, as in the cases of the other rules, it is
evident that $A$ is an instance of $A'$.

Suppose $A$ is derived from $B$ by weakening. If this is ogroup
weakening, $A$ is an instance of a binary tautology for the same reason as in the case
of exchange, duplication, $\vee$-introduction or
$\wedge$-introduction. Assume now we are dealing with pool
weakening, so that $A$ is the result of inserting a new oformula $F$
into $B$. By the induction hypothesis, $B$ is an instance of a binary
tautology $B'$. Let $P$ be a general atom not occurring in $B'$. And
let $A'$ be the result of applying weakening to $B'$ that inserts $P$
``at the same place" into $B'$ as the above application of weakening
inserted $F$ into $B$ when deriving $A$. Obviously $A'$ inherits binarity
from $B'$; by Lemma 1, it inherits from $B'$ tautologicity as well.
And, for the same reason as in all previous cases, $A$ is an instance of $A'$.

Finally, suppose $A$ is derived from $B$ by contraction. Then the contracted
formula $F$ should be elementary. By the induction hypothesis, $B$ is an
instance of a binary tautology $B'$. Let $\sigma$ be a substitution such that
$B=\sigma(B')$. And let $F'_1$, $F'_2$ be two oformulas in $B'$ ``at the same place" as
$F$, $F$ are in $B$, with $\sigma(F'_1)=F$ and $\sigma(F'_2)=F$.
Let $\delta$ be the substitution such that, for any general atom $P$, $\delta(P)=\sigma(P)$
if $P$ occurs in $F'_1$ or $F'_2$, and $\delta(P)=P$ otherwise. Thus, $\delta(F'_1)=\delta(F'_2)=F$.
And let $B''=\delta(B')$. Obviously --- for the same reasons as in classical logic --- substitution
does not destroy tautologicity, so $B''$ is a tautology because $B'$ is so. Further, the substitution
$\delta$ does not introduce any new occurrences of general atoms, so it does not destroy the
binarity of $B'$, either. To summarize, $B''$ is a binary tautology. Also, of course, $B$ is an
instance of $B''$. Notice that $B''$ has $F$ and $F$ where $B$ has the contracted oformulas $F$ and
$F$. So, let $A'$ be the result of applying contraction to $B''$ ``at the same place" as it
was applied to $B$ when deriving $A$ from it, as illustrated in the
following example:

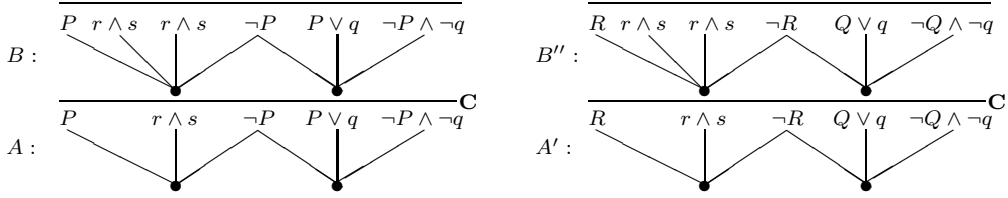
\begin{center}
\begin{picture}(80,90)(10,50)\footnotesize
\put(0,0){\begin{picture}(80,90)\put(-120,130){\line(1,0){152}}
\put(-120,120){$P\ \  r\wedge s\ \ \ r\wedge s\ \ \ \ \ \neg P\ \ \ \
P\vee q\ \ \ \neg P\wedge\neg q$}

\put(0,0){\begin{picture}(80,90)\put(-117,118){\line(2,-1){40}}\put(-97,118){\line(1,-1){20}}\put(-76,118){\line(0,-1){20}}\put(-45,118){\line(-3,-2){32}}

\put(-45,118){\line(3,-2){32}}\put(-15,118){\line(0,-1){20}}\put(18,118){\line(-5,-3){33}}\put(-76,97){\circle*{4}}\put(-15,97){\circle*{4}}\end{picture}}
\put(-120,93){\line(1,0){150}}\put(31,90){${\bf C}$}

\put(-120,84){$P\ \ \ \ \ \ \ \ \ \ r\wedge s\ \ \ \ \ \ \neg P\ \ \ \ P\vee q\ \ \ \neg
P\wedge\neg q$}

\put(0,-36){\begin{picture}(80,90)\put(-117,118){\line(2,-1){40}}\put(-76,118){\line(0,-1){20}}\put(-45,118){\line(-3,-2){32}}

\put(-45,118){\line(3,-2){32}}\put(-15,118){\line(0,-1){20}}\put(18,118){\line(-5,-3){33}}\put(-76,97){\circle*{4}}\put(-15,97){\circle*{4}}\end{picture}}

\put(-140,108){$B:$}\put(-140,73){$A:$}\end{picture}}

\put(200,0){\begin{picture}(80,90)\put(-120,130){\line(1,0){152}}
\put(-120,120){$R\ \ r\wedge s\ \ \ r\wedge s\ \ \ \ \neg R\ \
\ \ \ Q\vee q\ \ \ \neg Q\wedge\neg q$}

\put(0,0){\begin{picture}(80,90)\put(-117,118){\line(2,-1){40}}\put(-97,118){\line(1,-1){20}}\put(-76,118){\line(0,-1){20}}\put(-45,118){\line(-3,-2){32}}

\put(-45,118){\line(3,-2){32}}\put(-15,118){\line(0,-1){20}}\put(18,118){\line(-5,-3){33}}\put(-76,97){\circle*{4}}\put(-15,97){\circle*{4}}\end{picture}}

\put(-120,93){\line(1,0){150}}\put(31,90){${\bf C}$}

\put(-120,84){$R\ \ \ \ \ \ \ \ \ \ r\wedge s\ \ \ \ \ \ \neg R\ \ \ \ Q\vee q\ \
\ \neg Q\wedge\neg q$}

\put(0,-36){\begin{picture}(80,90)\put(-117,118){\line(2,-1){40}}\put(-76,118){\line(0,-1){20}}\put(-45,118){\line(-3,-2){32}}

\put(-45,118){\line(3,-2){32}}\put(-15,118){\line(0,-1){20}}\put(18,118){\line(-5,-3){33}}\put(-76,97){\circle*{4}}\put(-15,97){\circle*{4}}\end{picture}}
\put(-140,108){$B'':$}\put(-140,73){$A':$}\end{picture}}
\end{picture}
\end{center}
Obviously $A$ will be an instance of $A'$. And, by Lemma 1 and Lemma 6, $A'$ is a binary tautology.

$(\Leftarrow)$ Consider an arbitrary cirquent $A$ that is an
instance of a binary tautology $A'$. In view of Lemma 5, it would suffice to
show that {\bf CL6} proves $A'$. We construct a proof of $A'$, in the bottom-up fashion, as follows.
Starting from $A'$, we keep applying $\vee$-introduction and $\wedge$-introduction
until we hit an essentially literal cirquent\footnote[1]{An essentially literal cirquent, defined in \cite{Key1},
is one every oformula of whose pool either is an oliteral or is homeless.} $B$. As in the proof of Theorem 6 of \cite{Key1},
such a cirquent $B$ is guaranteed to be a tautology, and $A'$ follows from it in $(\vee\wedge)$.
Furthermore, in view of Lemma 3, $B$ is in fact a binary tautology.
The tautologicity of $B$ means that every ogroup of it contains either a $\top$, or at least one pair of opposite
(general or elementary) non-logical oliterals. For each ogroup of $B$ that contains a $\top$, pick one occurrence of
$\top$ and apply to $B$ a series of
weakenings to first delete all arcs but the arc pointing to the chosen occurrence, and next delete all
homeless oformulas if any such oformulas are present. For each ogroup of the resulting cirquent that contains
a pair of opposite non-logical oliterals, pick one such pair, and
continue applying a series of weakenings, as in the proof of Theorem 6 of \cite{Key1}, until a tautological cirquent
$C$ is hit with no homeless oformulas, where
every ogroup only has either a $\top$ or a pair of opposite non-logical oliterals.
By Lemma 4, $C$ remains binary.
Our target cirquent $A'$ is thus derivable from $C$ in (W$\vee\wedge$).
 Apply a series of contractions
to $C$ to separate all shared $\top$ and all shared elementary non-logical oliterals $p$ or $\neg p$, as illustrated below;
as a result, we get a cirquent $D$ which is still a binary tautology, but whose ogroups no longer share any elementary oformulas.
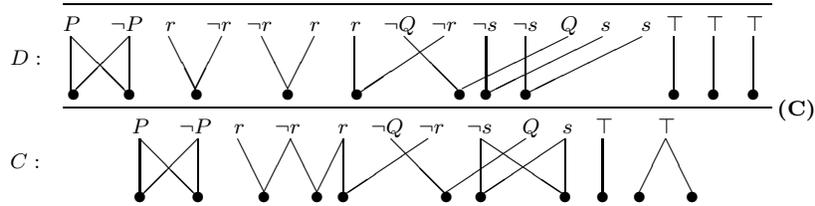
\begin{center}
\begin{picture}(80,90)(15,50)\footnotesize
\put(80,0){\begin{picture}(80,90)\put(-160,130){\line(1,0){268}}
\put(-160,120){$P\ \ \ \ \neg P\ \ \ r\ \ \ \ \neg r\ \ \neg r\ \ \ \ \ r\ \ \ \ r\ \ \ \neg Q\ \ \neg r\ \ \neg
s \ \ \neg s\ \ \ Q \ \ \ s\ \ \ \ s\ \ \top\ \ \ \top\ \ \ \top$}

\put(-120,118){\line(1,-2){11}}\put(-100,118){\line(-1,-2){11}}

\put(-31,118){\line(1,-1){21}}\put(-109.5,96){\circle*{4}}\put(-75,96){\circle*{4}}
\put(-85,118){\line(1,-2){10}}\put(-65,118){\line(-1,-2){10}}
\put(-49,96){\circle*{4}}\put(44,118){\line(-2,-1){43}}\put(0,96){\circle*{4}}\put(31,118){\line(-2,-1){40}}
\put(0,118){\line(0,-1){21}}\put(-10,96){\circle*{4}}
\put(-50,118){\line(0,-1){21}}\put(-16,118){\line(-3,-2){33}}
\put(15,96){\circle*{4}}\put(15,118){\line(0,-1){21}}\put(59,118){\line(-2,-1){43}}\put(71,118){\line(0,-1){21}}\put(71,96){\circle*{4}}
\put(86,118){\line(0,-1){21}}\put(86,96){\circle*{4}}\put(101,118){\line(0,-1){21}}\put(101,96){\circle*{4}}

\put(-160,91){\line(1,0){268}}\put(110,88){\bf (C)}
\put(-40,39){\begin{picture}(80,90)\put(-117,79){\line(0,-1){21}}\put(-95,79){\line(-1,-1){22}}\put(-95,79){\line(0,-1){21}}
\put(-95,57){\circle*{4}}\put(-117,79){\line(1,-1){22}}\put(-116,57){\circle*{4}}
\end{picture}}

\put(16,0){\begin{picture}(80,90)\put(-150,81){$P\ \ \ \ \neg P\ \ \ r\ \ \ \ \neg r\ \ \ \ \ r\ \ \ \neg Q\ \ \neg r\ \ \ \neg
s \ \ \ \ Q \ \ \ s\ \ \ \top\ \ \ \ \ \ \top$}

\put(-30,0){\begin{picture}(80,90)\put(-117,79){\line(0,-1){21}}\put(-95,79){\line(-1,-1){22}}\put(-95,79){\line(0,-1){21}}
\put(-95,57){\circle*{4}}\put(-117,79){\line(1,-1){22}}
\end{picture}}
\put(-52,79){\line(1,-1){21}}
\put(-147,57){\circle*{4}}
\put(-31,57){\circle*{4}}
\put(-110,79){\line(1,-2){10}}\put(-90,79){\line(-1,-2){10}}\put(-90,79){\line(1,-2){10}}\put(-70,79){\line(-1,-2){10}}
\put(-70,79){\line(0,-1){21}}\put(-70,57){\circle*{4}}\put(-80,57){\circle*{4}}\put(-38,79){\line(-3,-2){33}}
\put(-100,57){\circle*{4}}\put(-1,79){\line(-3,-2){30}}
\put(14,79){\line(-3,-2){33}}
\put(-18,79){\line(0,-1){21}}\put(-18,57){\circle*{4}}\put(14,79){\line(0,-1){21}}\put(14,57){\circle*{4}}\put(-18,79){\line(3,-2){33}}
\put(28,79){\line(0,-1){21}}\put(28,57){\circle*{4}}\put(52,79){\line(-1,-2){10}}\put(52,79){\line(1,-2){10}}\put(42,57){\circle*{4}}
\put(62,57){\circle*{4}}
\end{picture}}
\put(-180,107){$D:$}\put(-180,68){$C:$}\end{picture}}
\end{picture}
\end{center}
It is easy to see that the binarity of $D$ implies that there are no shared general oliterals $P$ or $\neg P$
in it except the cases when they are shared by identical-content ogroups.
Applying to $D$ a series of duplications, as illustrated below,
yields a cirquent $E$ that no longer has identical-content ogroups and hence no longer
has any shared oformulas.
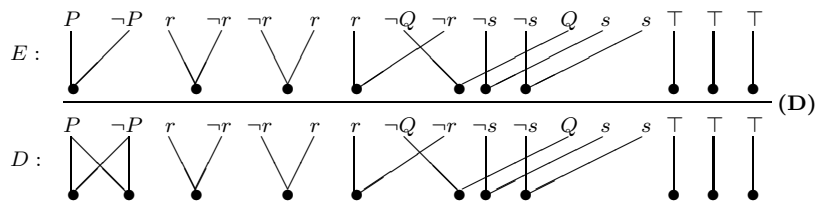
\begin{center}
\begin{picture}(80,90)(15,50)\footnotesize
\put(80,0){\begin{picture}(80,90)\put(-160,130){\line(1,0){268}}
\put(-160,120){$P\ \ \ \ \neg P\ \ \ r\ \ \ \ \neg r\ \ \neg r\ \ \ \ \ r\ \ \ \ r\ \ \ \neg Q\ \ \neg r\ \ \neg
s \ \ \neg s\ \ \ Q \ \ \ s\ \ \ \ s\ \ \top\ \ \ \top\ \ \ \top$}

\put(-120,118){\line(1,-2){11}}\put(-100,118){\line(-1,-2){11}}

\put(-31,118){\line(1,-1){21}}\put(-109.5,96){\circle*{4}}\put(-75,96){\circle*{4}}
\put(-85,118){\line(1,-2){10}}\put(-65,118){\line(-1,-2){10}}
\put(-49,96){\circle*{4}}\put(44,118){\line(-2,-1){43}}\put(0,96){\circle*{4}}\put(31,118){\line(-2,-1){40}}
\put(0,118){\line(0,-1){21}}\put(-10,96){\circle*{4}}
\put(-50,118){\line(0,-1){21}}\put(-16,118){\line(-3,-2){33}}
\put(15,96){\circle*{4}}\put(15,118){\line(0,-1){21}}\put(59,118){\line(-2,-1){43}}\put(71,118){\line(0,-1){21}}\put(71,96){\circle*{4}}
\put(101,118){\line(0,-1){21}}\put(101,96){\circle*{4}}\put(86,118){\line(0,-1){21}}\put(86,96){\circle*{4}}

\put(-160,91){\line(1,0){268}}\put(110,88){\bf (D)}
\put(-40,39){\begin{picture}(80,90)\put(-117,79){\line(0,-1){21}}\put(-95,79){\line(-1,-1){22}}
\put(-116,57){\circle*{4}}
\end{picture}}

\put(-180,107){$E:$}\end{picture}}
\put(80,-40){\begin{picture}(80,90)
\put(-160,120){$P\ \ \ \ \neg P\ \ \ r\ \ \ \ \neg r\ \ \neg r\ \ \ \ \ r\ \ \ \ r\ \ \ \neg Q\ \ \neg r\ \ \neg
s \ \ \neg s\ \ \ Q \ \ \ s\ \ \ \ s\ \ \top\ \ \ \top\ \ \ \top$}

\put(-120,118){\line(1,-2){11}}\put(-100,118){\line(-1,-2){11}}

\put(-31,118){\line(1,-1){21}}\put(-109.5,96){\circle*{4}}\put(-75,96){\circle*{4}}
\put(-85,118){\line(1,-2){10}}\put(-65,118){\line(-1,-2){10}}
\put(-49,96){\circle*{4}}\put(44,118){\line(-2,-1){43}}\put(0,96){\circle*{4}}\put(31,118){\line(-2,-1){40}}
\put(0,118){\line(0,-1){21}}\put(-10,96){\circle*{4}}
\put(-50,118){\line(0,-1){21}}\put(-16,118){\line(-3,-2){33}}
\put(15,96){\circle*{4}}\put(15,118){\line(0,-1){21}}\put(59,118){\line(-2,-1){43}}\put(71,118){\line(0,-1){21}}\put(71,96){\circle*{4}}

\put(-40,39){\begin{picture}(80,90)\put(-117,79){\line(0,-1){21}}\put(-95,79){\line(-1,-1){22}}\put(-95,79){\line(0,-1){21}}
\put(-95,57){\circle*{4}}\put(-117,79){\line(1,-1){22}}\put(-116,57){\circle*{4}}\put(126,79){\line(0,-1){21}}\put(126,57){\circle*{4}}
\put(141,79){\line(0,-1){21}}\put(141,57){\circle*{4}}
\end{picture}}

\put(-180,107){$D:$}\end{picture}}
\end{picture}
\end{center}
$A'$ is thus derivable from $E$ in (DCW$\vee\wedge$). In turn,
$E$ is obviously provable in (AME). So, {\bf CL6} proves $A'$.
\end{proof}

\begin{lemma}
A cirquent is an instance of a binary tautology iff it is an atomic-level instance of
some normal binary tautology.
\end{lemma}

\begin{proof}Our proof here almost literally follows the proof of Lemma 9 of \cite{Key1}.

The ``if" part is trivial. For the ``only if" part, assume $A$ is an instance of a
binary tautology $B$. Let $P_1,\ldots,P_n$ be all of the general atoms of $B$ that
have two positive or two negative occurrences in $B$. Let $Q_1,\ldots,Q_n$ be any pairwise
distinct general atoms not occurring in $B$. Let $C$ be the result of replacing in $B$
one of the two occurrences of $P_i$ by $Q_i$, for each $i=1,\ldots,n$. Then obviously $C$
is a normal binary cirquent, and $B$ an instance of it. By transitivity, $A$ (as an instance of $B$)
is also an instance of $C$.

We want to see that $C$ is a tautology. Deny this. Then there is a classical model $M$ in which
$C$ is false. Let $M'$ be the model such that:
\begin{itemize}
\item $M'$ agrees with $M$ on all atoms that are not among $P_1,\ldots,P_n,Q_1,\ldots,Q_n$;
\item for each $i\in\{1,\ldots,n\}$, $M'(P_i)=M'(Q_i)=${\itshape false} if $P_i$ and $Q_i$ are positive in $C$;
and $M'(P_i)=M'(Q_i)=${\itshape true} if $P_i$ and $Q_i$ are negative in $C$.
\end{itemize}
By induction on complexity, it can be easily seen that, for every subformula $F$
of a formula of $C$, whenever $F$ is false in $M$, so is it in $M'$.
This extends from (sub)formulas to groups of $C$ and hence $C$ itself. Thus $C$ is false
in $M'$ because it is false in $M$. But $M'$ does not distinguish between $P_i$ and $Q_i$
(any $1\leq i\leq n$). This clearly implies that $C$ and $B$ have the same truth value in $M'$.
That is, $B$ is false in $M'$, which is however impossible because $B$ is a tautology.
From this contradiction we conclude that $C$ is a (normal binary) tautology.

Let $\sigma$ be a substitution such that $A=\sigma(C)$. Let $\sigma'$ be a substitution such that,
for each general atom $P$ of $C$, $\sigma'(P)$ is the result of replacing in $\sigma(P)$ each occurrence
of each general atom by a new general atom in such a way that: no general atom occurs more than once in $\sigma'(P)$, and
whenever $P\neq Q$, no general atom occurs in both $\sigma'(P)$ and $\sigma'(Q)$.
Since $C$ is a binary tautology and is its own instance, by Lemma 7, {\bf CL6} proves $C$.
Then, by Lemma 5, {\bf CL6} proves $\sigma'(C)$ (an instance of $C$).
In view of Lemma 1, we immediately get that $\sigma'(C)$ is a tautology.
$\sigma'(C)$ can also be easily seen to be a normal binary cirquent, because $C$ is so.
Finally, with a little thought, $A$ can be seen to be an atomic-level instance of $\sigma'(C)$.
\end{proof}

\begin{lemma}
A {\bf CL6}-formula $F$ is provable in {\bf CL2} iff it is an
instance of a binary tautology.
\end{lemma}

\begin{proof}Again, it should be acknowledged that the present proof very closely follows
the proof of Lemma 27 of \cite{Key1}, even though there are certain differences.

$(\Rightarrow)$ Consider an arbitrary {\bf CL6}-formula $F$ provable
in {\bf CL2}. Fix a {\bf CL2}-proof of $F$ in the form of a sequence
$\langle F_n, F_{n-1}, \ldots, F_1\rangle$ of formulas, with $F_1 = F$. We may
assume that this sequence has no repetitions or other redundancies.
We claim that, for each $i$ with $1\leq i\leq n$, the following
conditions are satisfied:

{\bf Condition 1:} $F_i$ does not contain $\sqcap,\sqcup$.

{\bf Condition 2:} Whenever $F_i$ contains an elementary atom not occurring in $F$,
that atom is
non-logical, and has exactly two --- one positive and one negative
--- occurrences in $F_i$.

{\bf Condition 3:} If $i < n$, then $F_i$ is derived from
$F_{i+1}$ by Rule (c).

{\bf Condition 4:} $F_n$ is derived (from the empty set of
premises) by Rule (a).

Condition 4 is obvious, because it is only Rule (a) that may take
no premises. That Conditions 1-3 are also satisfied can be
verified by induction on $i$. For the basis case of $i =1$,
Conditions 1 and 2 are immediate. $F_1$ can not be derived by Rule (b)
because, by Condition 1, $F_1$ does not contain any $\sqcup$. Nor
can it be derived by Rule (a) unless $n = 1$, for otherwise either
$F_1$ would have to contain a $\sqcap$ (which is not the case
according to Condition $1$), or the proof of $F$ would have
redundancies as $F_1$ would not really need any premises. Thus, if
$1 < n$, the only possibility for $F_1$ is to be derived from $F_2$
by Rule (c). For the induction step, assume $i < n$ and the above
conditions are satisfied for $F_i$. According to Condition 3, $F_i$
is derived by Rule (c) from $F_{i+1}$ . This obviously implies that
$F_{i+1}$ inherits Conditions 1 and 2 from $F_i$. And that
Condition 3 also holds for $F_{i+1}$ can be shown in the same way
as we did for $F_1$.

As the conclusion of Rule (a) (Condition 4), $F_n$ is
stable. Let $G$ be the elementarization of $F_n$. The stability of
$F_n$ means that $G$ is a tautology. Let $H$ be the result of replacing in $G$
every
occurrence of $\top$ and $\bot$ (except those inherited from $F$)
by a general atom, in such a way that different occurrences of
$\top$, $\bot$ are replaced by different atoms. In view of Condition 2 (applied to $F_n$),
we see that, on top of these new general atoms and the elementary atoms inherited from $F$, the
only additional atoms that $H$ contains are elementary atoms with exactly two --- one positive and
one negative --- occurrences.
Let $H'$ be the result of replacing in $H$ every occurrence of every such
elementary atom by a general atom not occurring in $H$, in such a way that
different elementary atoms are replaced by different general atoms.
Then it is not hard to see that $H'$ is binary and $F_n$ is an instance of $H'$.
With Condition 3 in mind, by induction, one can further see that the formulas $F_{n-1}$, $F_{n-2}$, $\ldots$ are
also instances of $H'$. Thus, $F$ is an instance of $H'$.
It remains to show that $H'$ is a tautology. But this is indeed so because
$H'$ results from the tautological $G$ by replacing positive occurrences of $\bot$
and replacing two --- one positive and one negative
--- occurrences of elementary atoms by general atoms.
It is known from classical logic that such replacements do not destroy truth and hence
tautologicity of formulas.

$(\Leftarrow)$ Assume $F$ is a {\bf CL6}-formula which is an
instance of a binary tautology $T$. In view of Lemma 8, we may assume that $T$
is normal and $F$ is an atomic-level instance of it. Let us call the general atoms
that only have one occurrence in $T$ {\bf single}, and the general atoms that
have two occurrences {\bf married}. Let $\sigma$ be the substitution with $\sigma(T)=F$.
Let $G$ be the formula resulting from $T$ by the following steps:
substituting each single atom $P$ by $\sigma(P)$; substituting each married atom
$Q$ by $\sigma(Q)$ if $\sigma(Q)$ is elementary; substituting each married atom
$R$ by a non-logical elementary atom $r$ not occurring in $F$ if $\sigma(R)$ is general.
It is
clear that then $F$ can be derived from $G$ by a series of
applications of Rule (c), with each such application replacing two
--- a positive and a negative --- occurrences of some non-logical elementary
atom $r$ by $\sigma(R)$. So, in order to show that {\bf CL2} proves
$F$, it would suffice to verify that $G$ is stable and hence it can
be derived from the empty set of premises by Rule (a). But $G$ is
indeed stable. To see this, consider the elementarization $G'$ of
$G$. It results from $T$ by replacing the only occurrence of each single general atom by
some elementary atom, and doing the same with both occurrences of each married general atom.
In other words, $G'$ is an instance of $T$. Hence, as $T$ is a tautology, so is $G'$,
meaning that $G$ is stable.
\end{proof}

\begin{theorem}
A formula is provable in {\bf CL6} iff it is valid in computability
logic.
\end{theorem}

\begin{proof}
This theorem is an immediate corollary of Lemma 9, Lemma 7 and the
known fact (proven in \cite{Key2}) that {\bf CL2} is sound and
complete with respect to the semantics of computability logic.
\end{proof}



\end{document}